\documentclass{llncs}

\usepackage[utf8]{inputenc}

\usepackage{amssymb,amsmath}
\usepackage{graphicx}
\usepackage{hyperref}
\usepackage{mathpartir}
\usepackage[barr]{xy}
\usepackage{mdframed}
\usepackage{bbm}
\usepackage{wrapfig}

\let\mto\to                     
\let\to\relax                   
\newcommand{\to}{\rightarrow}

\newcommand{\rto}{\leftharpoonup}
\newcommand{\lto}{\rightharpoonup}
\newcommand{\LB}{\text{L}_!}
\newcommand{\LE}{\text{L}_\kappa}
\newcommand{\redto}{\rightsquigarrow}
\newcommand{\cat}[1]{\mathcal{#1}}
\newcommand{\Set}{\mathsf{Set}}
\newcommand{\Dial}[2]{\mathsf{Dial}_{#1}(#2)}
\newcommand{\Hom}[3]{\mathsf{Hom}_{#1}(#2,#3)}

\newcommand{\Ldrule}[4][]{{\displaystyle\frac{\begin{array}{l}#2\end{array}}{#3}\quad\Ldrulename{#4}}}

\newcommand{\Lpremise}[1]{ #1 \\}
\newenvironment{Ldefnblock}[3][]{ \framebox{\mbox{#2}} \quad #3 \\[0pt]}{}

\newcommand{\Lnt}[1]{\mathit{#1}}
\newcommand{\Lmv}[1]{\mathit{#1}}

\newcommand{\Lsym}[1]{#1}

\newcommand{\Ldrulename}[1]{\textsc{#1}}

\newcommand{\LdruleaxName}[0]{\Ldrulename{ax}}
\newcommand{\Ldruleax}[1]{\Ldrule[#1]{%
}{
 \Lnt{A}  \vdash  \Lnt{A} }{%
{\LdruleaxName}{}%
}}

\newcommand{\LdruleUrName}[0]{\Ldrulename{Ur}}
\newcommand{\LdruleUr}[1]{\Ldrule[#1]{%
}{
  \cdot   \vdash   I  }{%
{\LdruleUrName}{}%
}}

\newcommand{\LdrulecutName}[0]{\Ldrulename{cut}}
\newcommand{\Ldrulecut}[1]{\Ldrule[#1]{%
\Lpremise{  \Gamma_{{\mathrm{2}}}  \vdash  \Lnt{A}   \qquad   \Gamma_{{\mathrm{1}}}  \Lsym{,}  \Lnt{A}  \Lsym{,}  \Gamma_{{\mathrm{3}}}  \vdash  \Lnt{B}  }%
}{
 \Gamma_{{\mathrm{1}}}  \Lsym{,}  \Gamma_{{\mathrm{2}}}  \Lsym{,}  \Gamma_{{\mathrm{3}}}  \vdash  \Lnt{B} }{%
{\LdrulecutName}{}%
}}

\newcommand{\LdruleUlName}[0]{\Ldrulename{Ul}}
\newcommand{\LdruleUl}[1]{\Ldrule[#1]{%
\Lpremise{ \Gamma_{{\mathrm{1}}}  \Lsym{,}  \Gamma_{{\mathrm{2}}}  \vdash  \Lnt{A} }%
}{
 \Gamma_{{\mathrm{1}}}  \Lsym{,}   I   \Lsym{,}  \Gamma_{{\mathrm{2}}}  \vdash  \Lnt{A} }{%
{\LdruleUlName}{}%
}}

\newcommand{\LdruleTlName}[0]{\Ldrulename{Tl}}
\newcommand{\LdruleTl}[1]{\Ldrule[#1]{%
\Lpremise{ \Gamma  \Lsym{,}  \Lnt{A}  \Lsym{,}  \Lnt{B}  \Lsym{,}  \Gamma'  \vdash  \Lnt{C} }%
}{
 \Gamma  \Lsym{,}   \Lnt{A}  \otimes  \Lnt{B}   \Lsym{,}  \Gamma'  \vdash  \Lnt{C} }{%
{\LdruleTlName}{}%
}}

\newcommand{\LdruleTrName}[0]{\Ldrulename{Tr}}
\newcommand{\LdruleTr}[1]{\Ldrule[#1]{%
\Lpremise{  \Gamma_{{\mathrm{1}}}  \vdash  \Lnt{A}   \qquad   \Gamma_{{\mathrm{2}}}  \vdash  \Lnt{B}  }%
}{
 \Gamma_{{\mathrm{1}}}  \Lsym{,}  \Gamma_{{\mathrm{2}}}  \vdash   \Lnt{A}  \otimes  \Lnt{B}  }{%
{\LdruleTrName}{}%
}}

\newcommand{\LdruleIRlName}[0]{\Ldrulename{IRl}}
\newcommand{\LdruleIRl}[1]{\Ldrule[#1]{%
\Lpremise{  \Gamma_{{\mathrm{2}}}  \vdash  \Lnt{A}   \qquad   \Gamma_{{\mathrm{1}}}  \Lsym{,}  \Lnt{B}  \Lsym{,}  \Gamma_{{\mathrm{3}}}  \vdash  \Lnt{C}  }%
}{
 \Gamma_{{\mathrm{1}}}  \Lsym{,}   \Lnt{A}  \rightharpoonup  \Lnt{B}   \Lsym{,}  \Gamma_{{\mathrm{2}}}  \Lsym{,}  \Gamma_{{\mathrm{3}}}  \vdash  \Lnt{C} }{%
{\LdruleIRlName}{}%
}}

\newcommand{\LdruleILlName}[0]{\Ldrulename{ILl}}
\newcommand{\LdruleILl}[1]{\Ldrule[#1]{%
\Lpremise{  \Gamma_{{\mathrm{2}}}  \vdash  \Lnt{A}   \qquad   \Gamma_{{\mathrm{1}}}  \Lsym{,}  \Lnt{B}  \Lsym{,}  \Gamma_{{\mathrm{3}}}  \vdash  \Lnt{C}  }%
}{
 \Gamma_{{\mathrm{1}}}  \Lsym{,}  \Gamma_{{\mathrm{2}}}  \Lsym{,}   \Lnt{B}  \leftharpoonup  \Lnt{A}   \Lsym{,}  \Gamma_{{\mathrm{3}}}  \vdash  \Lnt{C} }{%
{\LdruleILlName}{}%
}}

\newcommand{\LdruleIRrName}[0]{\Ldrulename{IRr}}
\newcommand{\LdruleIRr}[1]{\Ldrule[#1]{%
\Lpremise{ \Gamma  \Lsym{,}  \Lnt{A}  \vdash  \Lnt{B} }%
}{
 \Gamma  \vdash   \Lnt{A}  \rightharpoonup  \Lnt{B}  }{%
{\LdruleIRrName}{}%
}}

\newcommand{\LdruleILrName}[0]{\Ldrulename{ILr}}
\newcommand{\LdruleILr}[1]{\Ldrule[#1]{%
\Lpremise{ \Lnt{A}  \Lsym{,}  \Gamma  \vdash  \Lnt{B} }%
}{
 \Gamma  \vdash   \Lnt{B}  \leftharpoonup  \Lnt{A}  }{%
{\LdruleILrName}{}%
}}

\newcommand{\LdruleCName}[0]{\Ldrulename{C}}
\newcommand{\LdruleC}[1]{\Ldrule[#1]{%
\Lpremise{ \Gamma_{{\mathrm{1}}}  \Lsym{,}   !  \Lnt{A}   \Lsym{,}  \Gamma_{{\mathrm{2}}}  \Lsym{,}   !  \Lnt{A}   \Lsym{,}  \Gamma_{{\mathrm{3}}}  \vdash  \Lnt{B} }%
}{
 \Gamma_{{\mathrm{1}}}  \Lsym{,}   !  \Lnt{A}   \Lsym{,}  \Gamma_{{\mathrm{2}}}  \Lsym{,}  \Gamma_{{\mathrm{3}}}  \vdash  \Lnt{B} }{%
{\LdruleCName}{}%
}}

\newcommand{\LdruleWName}[0]{\Ldrulename{W}}
\newcommand{\LdruleW}[1]{\Ldrule[#1]{%
\Lpremise{ \Gamma_{{\mathrm{1}}}  \Lsym{,}  \Gamma_{{\mathrm{2}}}  \vdash  \Lnt{B} }%
}{
 \Gamma_{{\mathrm{1}}}  \Lsym{,}   !  \Lnt{A}   \Lsym{,}  \Gamma_{{\mathrm{2}}}  \vdash  \Lnt{B} }{%
{\LdruleWName}{}%
}}

\newcommand{\LdruleBrName}[0]{\Ldrulename{Br}}
\newcommand{\LdruleBr}[1]{\Ldrule[#1]{%
\Lpremise{  !  \Gamma   \vdash  \Lnt{B} }%
}{
  !  \Gamma   \vdash   !  \Lnt{B}  }{%
{\LdruleBrName}{}%
}}

\newcommand{\LdruleBlName}[0]{\Ldrulename{Bl}}
\newcommand{\LdruleBl}[1]{\Ldrule[#1]{%
\Lpremise{ \Gamma_{{\mathrm{1}}}  \Lsym{,}  \Lnt{A}  \Lsym{,}  \Gamma_{{\mathrm{2}}}  \vdash  \Lnt{B} }%
}{
 \Gamma  \Lsym{,}   !  \Lnt{A}   \Lsym{,}  \Gamma_{{\mathrm{2}}}  \vdash  \Lnt{B} }{%
{\LdruleBlName}{}%
}}

\newcommand{\LdruleErName}[0]{\Ldrulename{Er}}
\newcommand{\LdruleEr}[1]{\Ldrule[#1]{%
\Lpremise{  \kappa  \Gamma   \vdash  \Lnt{B} }%
}{
  \kappa  \Gamma   \vdash   \kappa  \Lnt{B}  }{%
{\LdruleErName}{}%
}}

\newcommand{\LdruleElName}[0]{\Ldrulename{El}}
\newcommand{\LdruleEl}[1]{\Ldrule[#1]{%
\Lpremise{ \Gamma_{{\mathrm{1}}}  \Lsym{,}  \Lnt{A}  \Lsym{,}  \Gamma_{{\mathrm{2}}}  \vdash  \Lnt{B} }%
}{
 \Gamma_{{\mathrm{1}}}  \Lsym{,}   \kappa  \Lnt{A}   \Lsym{,}  \Gamma_{{\mathrm{2}}}  \vdash  \Lnt{B} }{%
{\LdruleElName}{}%
}}

\newcommand{\LdruleEOneName}[0]{\Ldrulename{E1}}
\newcommand{\LdruleEOne}[1]{\Ldrule[#1]{%
\Lpremise{ \Gamma_{{\mathrm{1}}}  \Lsym{,}   \kappa  \Lnt{A}   \Lsym{,}  \Lnt{B}  \Lsym{,}  \Gamma_{{\mathrm{2}}}  \vdash  \Lnt{C} }%
}{
 \Gamma_{{\mathrm{1}}}  \Lsym{,}  \Lnt{B}  \Lsym{,}   \kappa  \Lnt{A}   \Lsym{,}  \Gamma_{{\mathrm{2}}}  \vdash  \Lnt{C} }{%
{\LdruleEOneName}{}%
}}

\newcommand{\LdruleETwoName}[0]{\Ldrulename{E2}}
\newcommand{\LdruleETwo}[1]{\Ldrule[#1]{%
\Lpremise{ \Gamma_{{\mathrm{1}}}  \Lsym{,}  \Lnt{A}  \Lsym{,}   \kappa  \Lnt{B}   \Lsym{,}  \Gamma_{{\mathrm{2}}}  \vdash  \Lnt{C} }%
}{
 \Gamma_{{\mathrm{1}}}  \Lsym{,}   \kappa  \Lnt{B}   \Lsym{,}  \Lnt{A}  \Lsym{,}  \Gamma_{{\mathrm{2}}}  \vdash  \Lnt{C} }{%
{\LdruleETwoName}{}%
}}

\newcommand{\LdruleTXXvarName}[0]{\Ldrulename{T\_var}}
\newcommand{\LdruleTXXvar}[1]{\Ldrule[#1]{%
}{
 \Lmv{x}  \Lsym{:}  \Lnt{A}  \vdash  \Lmv{x}  :  \Lnt{A} }{%
{\LdruleTXXvarName}{}%
}}

\newcommand{\LdruleTXXUrName}[0]{\Ldrulename{T\_Ur}}
\newcommand{\LdruleTXXUr}[1]{\Ldrule[#1]{%
}{
  \cdot   \vdash   \mathsf{unit}   :   I  }{%
{\LdruleTXXUrName}{}%
}}

\newcommand{\LdruleTXXcutName}[0]{\Ldrulename{T\_cut}}
\newcommand{\LdruleTXXcut}[1]{\Ldrule[#1]{%
\Lpremise{  \Gamma_{{\mathrm{2}}}  \vdash  \Lnt{t_{{\mathrm{1}}}}  :  \Lnt{A}   \qquad   \Gamma_{{\mathrm{1}}}  \Lsym{,}  \Lmv{x}  \Lsym{:}  \Lnt{A}  \Lsym{,}  \Gamma_{{\mathrm{3}}}  \vdash  \Lnt{t_{{\mathrm{2}}}}  :  \Lnt{B}  }%
}{
 \Gamma_{{\mathrm{1}}}  \Lsym{,}  \Gamma_{{\mathrm{2}}}  \Lsym{,}  \Gamma_{{\mathrm{3}}}  \vdash  \Lsym{[}  \Lnt{t_{{\mathrm{1}}}}  \Lsym{/}  \Lmv{x}  \Lsym{]}  \Lnt{t_{{\mathrm{2}}}}  :  \Lnt{B} }{%
{\LdruleTXXcutName}{}%
}}

\newcommand{\LdruleTXXUlName}[0]{\Ldrulename{T\_Ul}}
\newcommand{\LdruleTXXUl}[1]{\Ldrule[#1]{%
\Lpremise{ \Gamma_{{\mathrm{1}}}  \Lsym{,}  \Gamma_{{\mathrm{2}}}  \vdash  \Lnt{t}  :  \Lnt{A} }%
}{
 \Gamma_{{\mathrm{1}}}  \Lsym{,}  \Lmv{x}  \Lsym{:}   I   \Lsym{,}  \Gamma_{{\mathrm{2}}}  \vdash   \mathsf{let}\, \Lmv{x} \,\mathsf{be}\,  \mathsf{unit}  \,\mathsf{in}\, \Lnt{t}   :  \Lnt{A} }{%
{\LdruleTXXUlName}{}%
}}

\newcommand{\LdruleTXXTlName}[0]{\Ldrulename{T\_Tl}}
\newcommand{\LdruleTXXTl}[1]{\Ldrule[#1]{%
\Lpremise{ \Gamma  \Lsym{,}  \Lmv{x}  \Lsym{:}  \Lnt{A}  \Lsym{,}  \Lmv{y}  \Lsym{:}  \Lnt{B}  \Lsym{,}  \Gamma'  \vdash  \Lnt{t}  :  \Lnt{C} }%
}{
 \Gamma  \Lsym{,}  \Lmv{z}  \Lsym{:}   \Lnt{A}  \otimes  \Lnt{B}   \Lsym{,}  \Gamma'  \vdash   \mathsf{let}\, \Lmv{z} \,\mathsf{be}\,  \Lmv{x}  \otimes  \Lmv{y}  \,\mathsf{in}\, \Lnt{t}   :  \Lnt{C} }{%
{\LdruleTXXTlName}{}%
}}

\newcommand{\LdruleTXXTrName}[0]{\Ldrulename{T\_Tr}}
\newcommand{\LdruleTXXTr}[1]{\Ldrule[#1]{%
\Lpremise{  \Gamma_{{\mathrm{1}}}  \vdash  \Lnt{t_{{\mathrm{1}}}}  :  \Lnt{A}   \qquad   \Gamma_{{\mathrm{2}}}  \vdash  \Lnt{t_{{\mathrm{2}}}}  :  \Lnt{B}  }%
}{
 \Gamma_{{\mathrm{1}}}  \Lsym{,}  \Gamma_{{\mathrm{2}}}  \vdash    \Lnt{t_{{\mathrm{1}}}}  \otimes  \Lnt{t_{{\mathrm{2}}}}    :    \Lnt{A}  \otimes  \Lnt{B}   }{%
{\LdruleTXXTrName}{}%
}}

\newcommand{\LdruleTXXIRlName}[0]{\Ldrulename{T\_IRl}}
\newcommand{\LdruleTXXIRl}[1]{\Ldrule[#1]{%
\Lpremise{  \Gamma_{{\mathrm{2}}}  \vdash  \Lnt{t_{{\mathrm{1}}}}  :  \Lnt{A}   \qquad   \Gamma_{{\mathrm{1}}}  \Lsym{,}  \Lmv{x}  \Lsym{:}  \Lnt{B}  \Lsym{,}  \Gamma_{{\mathrm{3}}}  \vdash  \Lnt{t_{{\mathrm{2}}}}  :  \Lnt{C}  }%
}{
 \Gamma_{{\mathrm{1}}}  \Lsym{,}  \Lmv{z}  \Lsym{:}   \Lnt{A}  \rightharpoonup  \Lnt{B}   \Lsym{,}  \Gamma_{{\mathrm{2}}}  \Lsym{,}  \Gamma_{{\mathrm{3}}}  \vdash  \Lsym{[}   \mathsf{app}_r\, \Lmv{z} \, \Lnt{t_{{\mathrm{1}}}}   \Lsym{/}  \Lmv{x}  \Lsym{]}  \Lnt{t_{{\mathrm{2}}}}  :  \Lnt{C} }{%
{\LdruleTXXIRlName}{}%
}}

\newcommand{\LdruleTXXILlName}[0]{\Ldrulename{T\_ILl}}
\newcommand{\LdruleTXXILl}[1]{\Ldrule[#1]{%
\Lpremise{  \Gamma_{{\mathrm{2}}}  \vdash  \Lnt{t_{{\mathrm{1}}}}  :  \Lnt{A}   \qquad   \Gamma_{{\mathrm{1}}}  \Lsym{,}  \Lmv{x}  \Lsym{:}  \Lnt{B}  \Lsym{,}  \Gamma_{{\mathrm{3}}}  \vdash  \Lnt{t_{{\mathrm{2}}}}  :  \Lnt{C}  }%
}{
 \Gamma_{{\mathrm{1}}}  \Lsym{,}  \Gamma_{{\mathrm{2}}}  \Lsym{,}  \Lmv{z}  \Lsym{:}   \Lnt{B}  \leftharpoonup  \Lnt{A}   \Lsym{,}  \Gamma_{{\mathrm{3}}}  \vdash  \Lsym{[}   \mathsf{app}_l\, \Lmv{z} \, \Lnt{t_{{\mathrm{1}}}}   \Lsym{/}  \Lmv{x}  \Lsym{]}  \Lnt{t_{{\mathrm{2}}}}  :  \Lnt{C} }{%
{\LdruleTXXILlName}{}%
}}

\newcommand{\LdruleTXXIlName}[0]{\Ldrulename{T\_Il}}
\newcommand{\LdruleTXXIl}[1]{\Ldrule[#1]{%
\Lpremise{  \Gamma_{{\mathrm{2}}}  \vdash  \Lnt{t_{{\mathrm{1}}}}  :  \Lnt{A}   \qquad   \Gamma_{{\mathrm{1}}}  \Lsym{,}  \Lmv{x}  \Lsym{:}  \Lnt{B}  \Lsym{,}  \Gamma_{{\mathrm{3}}}  \vdash  \Lnt{t_{{\mathrm{2}}}}  :  \Lnt{C}  }%
}{
 \Gamma_{{\mathrm{1}}}  \Lsym{,}  \Lmv{z}  \Lsym{:}   \Lnt{A}  \multimap  \Lnt{B}   \Lsym{,}  \Gamma_{{\mathrm{2}}}  \Lsym{,}  \Gamma_{{\mathrm{3}}}  \vdash  \Lsym{[}   \Lmv{z} \, \Lnt{t_{{\mathrm{1}}}}   \Lsym{/}  \Lmv{x}  \Lsym{]}  \Lnt{t_{{\mathrm{2}}}}  :  \Lnt{C} }{%
{\LdruleTXXIlName}{}%
}}

\newcommand{\LdruleTXXIRrName}[0]{\Ldrulename{T\_IRr}}
\newcommand{\LdruleTXXIRr}[1]{\Ldrule[#1]{%
\Lpremise{ \Gamma  \Lsym{,}  \Lmv{x}  \Lsym{:}  \Lnt{A}  \vdash  \Lnt{t}  :  \Lnt{B} }%
}{
 \Gamma  \vdash   \lambda_r  \Lmv{x} : \Lnt{A} . \Lnt{t}   :   \Lnt{A}  \rightharpoonup  \Lnt{B}  }{%
{\LdruleTXXIRrName}{}%
}}

\newcommand{\LdruleTXXILrName}[0]{\Ldrulename{T\_ILr}}
\newcommand{\LdruleTXXILr}[1]{\Ldrule[#1]{%
\Lpremise{ \Lmv{x}  \Lsym{:}  \Lnt{A}  \Lsym{,}  \Gamma  \vdash  \Lnt{t}  :  \Lnt{B} }%
}{
 \Gamma  \vdash   \lambda_l  \Lmv{x} : \Lnt{A} . \Lnt{t}   :   \Lnt{B}  \leftharpoonup  \Lnt{A}  }{%
{\LdruleTXXILrName}{}%
}}

\newcommand{\LdruleTXXIrName}[0]{\Ldrulename{T\_Ir}}
\newcommand{\LdruleTXXIr}[1]{\Ldrule[#1]{%
\Lpremise{ \Gamma  \Lsym{,}  \Lmv{x}  \Lsym{:}  \Lnt{A}  \vdash  \Lnt{t}  :  \Lnt{B} }%
}{
 \Gamma  \vdash   \lambda  \Lmv{x} : \Lnt{A} . \Lnt{t}   :   \Lnt{A}  \multimap  \Lnt{B}  }{%
{\LdruleTXXIrName}{}%
}}

\newcommand{\LdruleTXXCName}[0]{\Ldrulename{T\_C}}
\newcommand{\LdruleTXXC}[1]{\Ldrule[#1]{%
\Lpremise{ \Gamma_{{\mathrm{1}}}  \Lsym{,}  \Lmv{x}  \Lsym{:}   !  \Lnt{A}   \Lsym{,}  \Gamma_{{\mathrm{2}}}  \Lsym{,}  \Lmv{y}  \Lsym{:}   !  \Lnt{A}   \Lsym{,}  \Gamma_{{\mathrm{3}}}  \vdash  \Lnt{t}  :  \Lnt{B} }%
}{
 \Gamma_{{\mathrm{1}}}  \Lsym{,}  \Lmv{z}  \Lsym{:}   !  \Lnt{A}   \Lsym{,}  \Gamma_{{\mathrm{2}}}  \Lsym{,}  \Gamma_{{\mathrm{3}}}  \vdash   \mathsf{copy}\, \Lmv{z} \,\mathsf{as}\, \Lmv{x} , \Lmv{y} \,\mathsf{in}\, \Lnt{t}   :  \Lnt{B} }{%
{\LdruleTXXCName}{}%
}}

\newcommand{\LdruleTXXWName}[0]{\Ldrulename{T\_W}}
\newcommand{\LdruleTXXW}[1]{\Ldrule[#1]{%
\Lpremise{ \Gamma_{{\mathrm{1}}}  \Lsym{,}  \Gamma_{{\mathrm{2}}}  \vdash  \Lnt{t}  :  \Lnt{B} }%
}{
 \Gamma_{{\mathrm{1}}}  \Lsym{,}  \Lmv{x}  \Lsym{:}   !  \Lnt{A}   \Lsym{,}  \Gamma_{{\mathrm{2}}}  \vdash   \mathsf{discard}\, \Lmv{x} \,\mathsf{in}\, \Lnt{t}   :  \Lnt{B} }{%
{\LdruleTXXWName}{}%
}}

\newcommand{\LdruleTXXBrName}[0]{\Ldrulename{T\_Br}}
\newcommand{\LdruleTXXBr}[1]{\Ldrule[#1]{%
\Lpremise{  \overrightarrow{ \Lmv{x} } : ! \Gamma   \vdash  \Lnt{t}  :  \Lnt{B} }%
}{
  \overrightarrow{ \Lmv{y} } : ! \Gamma   \vdash   \mathsf{promote}_!\,  \overrightarrow{ \Lmv{y} }  \,\mathsf{for}\,  \overrightarrow{ \Lmv{x} }  \,\mathsf{in}\, \Lnt{t}   :   !  \Lnt{B}  }{%
{\LdruleTXXBrName}{}%
}}

\newcommand{\LdruleTXXBlName}[0]{\Ldrulename{T\_Bl}}
\newcommand{\LdruleTXXBl}[1]{\Ldrule[#1]{%
\Lpremise{ \Gamma_{{\mathrm{1}}}  \Lsym{,}  \Lmv{x}  \Lsym{:}  \Lnt{A}  \Lsym{,}  \Gamma_{{\mathrm{2}}}  \vdash  \Lnt{t}  :  \Lnt{B} }%
}{
 \Gamma_{{\mathrm{1}}}  \Lsym{,}  \Lmv{y}  \Lsym{:}   !  \Lnt{A}   \Lsym{,}  \Gamma_{{\mathrm{2}}}  \vdash  \Lsym{[}   \mathsf{derelict}_!\, \Lmv{y}   \Lsym{/}  \Lmv{x}  \Lsym{]}  \Lnt{t}  :  \Lnt{B} }{%
{\LdruleTXXBlName}{}%
}}

\newcommand{\LdruleTXXErName}[0]{\Ldrulename{T\_Er}}
\newcommand{\LdruleTXXEr}[1]{\Ldrule[#1]{%
\Lpremise{  \overrightarrow{ \Lmv{x} } : \kappa  \Gamma   \vdash  \Lnt{t}  :  \Lnt{B} }%
}{
  \overrightarrow{ \Lmv{y} } : \kappa  \Gamma   \vdash   \mathsf{promote}_\kappa\,  \overrightarrow{ \Lmv{y} }  \,\mathsf{for}\,  \overrightarrow{ \Lmv{x} }  \,\mathsf{in}\, \Lnt{t}   :   \kappa  \Lnt{B}  }{%
{\LdruleTXXErName}{}%
}}

\newcommand{\LdruleTXXElName}[0]{\Ldrulename{T\_El}}
\newcommand{\LdruleTXXEl}[1]{\Ldrule[#1]{%
\Lpremise{ \Gamma_{{\mathrm{1}}}  \Lsym{,}  \Lmv{x}  \Lsym{:}  \Lnt{A}  \Lsym{,}  \Gamma_{{\mathrm{2}}}  \vdash  \Lnt{t}  :  \Lnt{B} }%
}{
 \Gamma_{{\mathrm{1}}}  \Lsym{,}  \Lmv{y}  \Lsym{:}   \kappa  \Lnt{A}   \Lsym{,}  \Gamma_{{\mathrm{2}}}  \vdash  \Lsym{[}   \mathsf{derelict}_\kappa\, \Lmv{y}   \Lsym{/}  \Lmv{x}  \Lsym{]}  \Lnt{t}  :  \Lnt{B} }{%
{\LdruleTXXElName}{}%
}}

\newcommand{\LdruleTXXEOneName}[0]{\Ldrulename{T\_E1}}
\newcommand{\LdruleTXXEOne}[1]{\Ldrule[#1]{%
\Lpremise{ \Gamma_{{\mathrm{1}}}  \Lsym{,}  \Lmv{x_{{\mathrm{1}}}}  \Lsym{:}   \kappa  \Lnt{A}   \Lsym{,}  \Lmv{y_{{\mathrm{1}}}}  \Lsym{:}  \Lnt{B}  \Lsym{,}  \Gamma_{{\mathrm{2}}}  \vdash  \Lnt{t}  :  \Lnt{C} }%
}{
 \Gamma_{{\mathrm{1}}}  \Lsym{,}  \Lmv{y_{{\mathrm{2}}}}  \Lsym{:}  \Lnt{B}  \Lsym{,}  \Lmv{x_{{\mathrm{2}}}}  \Lsym{:}   \kappa  \Lnt{A}   \Lsym{,}  \Gamma_{{\mathrm{2}}}  \vdash   \mathsf{exchange_l}\, \Lmv{y_{{\mathrm{2}}}} , \Lmv{x_{{\mathrm{2}}}} \,\mathsf{with}\, \Lmv{x_{{\mathrm{1}}}} , \Lmv{y_{{\mathrm{1}}}} \,\mathsf{in}\, \Lnt{t}   :  \Lnt{C} }{%
{\LdruleTXXEOneName}{}%
}}

\newcommand{\LdruleTXXETwoName}[0]{\Ldrulename{T\_E2}}
\newcommand{\LdruleTXXETwo}[1]{\Ldrule[#1]{%
\Lpremise{ \Gamma_{{\mathrm{1}}}  \Lsym{,}  \Lmv{x_{{\mathrm{1}}}}  \Lsym{:}  \Lnt{A}  \Lsym{,}  \Lmv{y_{{\mathrm{1}}}}  \Lsym{:}   \kappa  \Lnt{B}   \Lsym{,}  \Gamma_{{\mathrm{2}}}  \vdash  \Lnt{t}  :  \Lnt{C} }%
}{
 \Gamma_{{\mathrm{1}}}  \Lsym{,}  \Lmv{y_{{\mathrm{2}}}}  \Lsym{:}   \kappa  \Lnt{B}   \Lsym{,}  \Lmv{x_{{\mathrm{2}}}}  \Lsym{:}  \Lnt{A}  \Lsym{,}  \Gamma_{{\mathrm{2}}}  \vdash   \mathsf{exchange_r}\, \Lmv{y_{{\mathrm{2}}}} , \Lmv{x_{{\mathrm{2}}}} \,\mathsf{with}\, \Lmv{x_{{\mathrm{1}}}} , \Lmv{y_{{\mathrm{1}}}} \,\mathsf{in}\, \Lnt{t}   :  \Lnt{C} }{%
{\LdruleTXXETwoName}{}%
}}

\newcommand{\LdruleRXXBetalName}[0]{\Ldrulename{R\_Betal}}
\newcommand{\LdruleRXXBetal}[1]{\Ldrule[#1]{%
}{
  \mathsf{app}_l\,  (   \lambda_l  \Lmv{x} : \Lnt{A} . \Lnt{t_{{\mathrm{2}}}}   )  \, \Lnt{t_{{\mathrm{1}}}}   \rightsquigarrow  \Lsym{[}  \Lnt{t_{{\mathrm{1}}}}  \Lsym{/}  \Lmv{x}  \Lsym{]}  \Lnt{t_{{\mathrm{2}}}} }{%
{\LdruleRXXBetalName}{}%
}}

\newcommand{\LdruleRXXBetarName}[0]{\Ldrulename{R\_Betar}}
\newcommand{\LdruleRXXBetar}[1]{\Ldrule[#1]{%
}{
  \mathsf{app}_r\,  (   \lambda_r  \Lmv{x} : \Lnt{A} . \Lnt{t_{{\mathrm{2}}}}   )  \, \Lnt{t_{{\mathrm{1}}}}   \rightsquigarrow  \Lsym{[}  \Lnt{t_{{\mathrm{1}}}}  \Lsym{/}  \Lmv{x}  \Lsym{]}  \Lnt{t_{{\mathrm{2}}}} }{%
{\LdruleRXXBetarName}{}%
}}

\newcommand{\LdruleRXXBetaUName}[0]{\Ldrulename{R\_BetaU}}
\newcommand{\LdruleRXXBetaU}[1]{\Ldrule[#1]{%
}{
  \mathsf{let}\, \Lnt{t_{{\mathrm{1}}}} \,\mathsf{be}\,  \mathsf{unit}  \,\mathsf{in}\, \Lsym{[}   \mathsf{unit}   \Lsym{/}  \Lmv{z}  \Lsym{]}  \Lnt{t_{{\mathrm{2}}}}   \rightsquigarrow  \Lsym{[}  \Lnt{t_{{\mathrm{1}}}}  \Lsym{/}  \Lmv{z}  \Lsym{]}  \Lnt{t_{{\mathrm{2}}}} }{%
{\LdruleRXXBetaUName}{}%
}}

\newcommand{\LdruleRXXBetaTOneName}[0]{\Ldrulename{R\_BetaT1}}
\newcommand{\LdruleRXXBetaTOne}[1]{\Ldrule[#1]{%
}{
  \mathsf{let}\,  \Lnt{t_{{\mathrm{1}}}}  \otimes  \Lnt{t_{{\mathrm{2}}}}  \,\mathsf{be}\,  \Lmv{x}  \otimes  \Lmv{y}  \,\mathsf{in}\, \Lnt{t}   \rightsquigarrow  \Lsym{[}  \Lnt{t_{{\mathrm{1}}}}  \Lsym{/}  \Lmv{x}  \Lsym{]}  \Lsym{[}  \Lnt{t_{{\mathrm{2}}}}  \Lsym{/}  \Lmv{y}  \Lsym{]}  \Lnt{t} }{%
{\LdruleRXXBetaTOneName}{}%
}}

\newcommand{\LdruleRXXBetaTTwoName}[0]{\Ldrulename{R\_BetaT2}}
\newcommand{\LdruleRXXBetaTTwo}[1]{\Ldrule[#1]{%
}{
  \mathsf{let}\, \Lnt{t_{{\mathrm{1}}}} \,\mathsf{be}\,  \Lmv{x}  \otimes  \Lmv{y}  \,\mathsf{in}\, \Lsym{[}   \Lmv{x}  \otimes  \Lmv{y}   \Lsym{/}  \Lmv{z}  \Lsym{]}  \Lnt{t_{{\mathrm{2}}}}   \rightsquigarrow  \Lsym{[}  \Lnt{t_{{\mathrm{1}}}}  \Lsym{/}  \Lmv{x}  \Lsym{]}  \Lnt{t_{{\mathrm{2}}}} }{%
{\LdruleRXXBetaTTwoName}{}%
}}

\newcommand{\LdruleRXXNatUName}[0]{\Ldrulename{R\_NatU}}
\newcommand{\LdruleRXXNatU}[1]{\Ldrule[#1]{%
}{
 \Lsym{[}   \mathsf{let}\, \Lnt{t_{{\mathrm{1}}}} \,\mathsf{be}\,  \mathsf{unit}  \,\mathsf{in}\, \Lnt{t_{{\mathrm{2}}}}   \Lsym{/}  \Lmv{z}  \Lsym{]}  \Lnt{t_{{\mathrm{3}}}}  \rightsquigarrow   \mathsf{let}\, \Lnt{t_{{\mathrm{1}}}} \,\mathsf{be}\,  \mathsf{unit}  \,\mathsf{in}\, \Lsym{[}  \Lnt{t_{{\mathrm{2}}}}  \Lsym{/}  \Lmv{z}  \Lsym{]}  \Lnt{t_{{\mathrm{3}}}}  }{%
{\LdruleRXXNatUName}{}%
}}

\newcommand{\LdruleRXXNatTName}[0]{\Ldrulename{R\_NatT}}
\newcommand{\LdruleRXXNatT}[1]{\Ldrule[#1]{%
}{
 \Lsym{[}   \mathsf{let}\, \Lnt{t_{{\mathrm{1}}}} \,\mathsf{be}\,  \Lmv{x}  \otimes  \Lmv{y}  \,\mathsf{in}\, \Lnt{t_{{\mathrm{2}}}}   \Lsym{/}  \Lmv{z}  \Lsym{]}  \Lnt{t_{{\mathrm{3}}}}  \rightsquigarrow   \mathsf{let}\, \Lnt{t_{{\mathrm{1}}}} \,\mathsf{be}\,  \Lmv{x}  \otimes  \Lmv{y}  \,\mathsf{in}\, \Lsym{[}  \Lnt{t_{{\mathrm{2}}}}  \Lsym{/}  \Lmv{z}  \Lsym{]}  \Lnt{t_{{\mathrm{3}}}}  }{%
{\LdruleRXXNatTName}{}%
}}

\newcommand{\LdruleRXXLetUName}[0]{\Ldrulename{R\_LetU}}
\newcommand{\LdruleRXXLetU}[1]{\Ldrule[#1]{%
}{
  \mathsf{let}\,  \mathsf{unit}  \,\mathsf{be}\,  \mathsf{unit}  \,\mathsf{in}\, \Lnt{t}   \rightsquigarrow  \Lnt{t} }{%
{\LdruleRXXLetUName}{}%
}}

\newcommand{\LdruleRXXBetaDRName}[0]{\Ldrulename{R\_BetaDR}}
\newcommand{\LdruleRXXBetaDR}[1]{\Ldrule[#1]{%
}{
  \mathsf{derelict}_!\,  (   \mathsf{promote}_!\,  \overrightarrow{ \Lnt{t} }  \,\mathsf{for}\,  \overrightarrow{ \Lmv{x} }  \,\mathsf{in}\, \Lnt{t_{{\mathrm{1}}}}   )    \rightsquigarrow  \Lsym{[}   \overrightarrow{ \Lnt{t} }   \Lsym{/}   \overrightarrow{ \Lmv{x} }   \Lsym{]}  \Lnt{t_{{\mathrm{1}}}} }{%
{\LdruleRXXBetaDRName}{}%
}}

\newcommand{\LdruleRXXBetaDIName}[0]{\Ldrulename{R\_BetaDI}}
\newcommand{\LdruleRXXBetaDI}[1]{\Ldrule[#1]{%
}{
  \mathsf{discard}\,  (   \mathsf{promote}_!\,  \overrightarrow{ \Lnt{t} }  \,\mathsf{for}\,  \overrightarrow{ \Lmv{x} }  \,\mathsf{in}\, \Lnt{t_{{\mathrm{1}}}}   )  \,\mathsf{in}\, \Lnt{t_{{\mathrm{2}}}}   \rightsquigarrow   \mathsf{discard}\,  \overrightarrow{ \Lnt{t} }  \,\mathsf{in}\, \Lnt{t_{{\mathrm{2}}}}  }{%
{\LdruleRXXBetaDIName}{}%
}}

\newcommand{\LdruleRXXBetaCName}[0]{\Ldrulename{R\_BetaC}}
\newcommand{\LdruleRXXBetaC}[1]{\Ldrule[#1]{%
\Lpremise{  \Lnt{t'_{{\mathrm{1}}}}  \Lsym{=}   \mathsf{promote}_!\,  \overrightarrow{ \Lmv{w} }  \,\mathsf{for}\,  \overrightarrow{ \Lmv{x} }  \,\mathsf{in}\, \Lnt{t_{{\mathrm{1}}}}    \qquad   \Lnt{t''_{{\mathrm{1}}}}  \Lsym{=}   \mathsf{promote}_!\,  \overrightarrow{ \Lmv{z} }  \,\mathsf{for}\,  \overrightarrow{ \Lmv{x} }  \,\mathsf{in}\, \Lnt{t_{{\mathrm{1}}}}   }%
}{
  \mathsf{copy}\,  (   \mathsf{promote}_!\,  \overrightarrow{ \Lnt{t} }  \,\mathsf{for}\,  \overrightarrow{ \Lmv{x} }  \,\mathsf{in}\, \Lnt{t_{{\mathrm{1}}}}   )  \,\mathsf{as}\, \Lmv{w} , \Lmv{z} \,\mathsf{in}\, \Lnt{t_{{\mathrm{2}}}}   \rightsquigarrow   \mathsf{copy}\,  \overrightarrow{ \Lnt{t} }  \,\mathsf{as}\,  \overrightarrow{ \Lmv{w} }  ,  \overrightarrow{ \Lmv{z} }  \,\mathsf{in}\, \Lsym{[}  \Lnt{t'_{{\mathrm{1}}}}  \Lsym{/}  \Lmv{w}  \Lsym{]}  \Lsym{[}  \Lnt{t''_{{\mathrm{1}}}}  \Lsym{/}  \Lmv{z}  \Lsym{]}  \Lnt{t_{{\mathrm{2}}}}  }{%
{\LdruleRXXBetaCName}{}%
}}

\newcommand{\LdruleRXXNatDName}[0]{\Ldrulename{R\_NatD}}
\newcommand{\LdruleRXXNatD}[1]{\Ldrule[#1]{%
}{
 \Lsym{[}   \mathsf{discard}\, \Lnt{t} \,\mathsf{in}\, \Lnt{t_{{\mathrm{1}}}}   \Lsym{/}  \Lmv{x}  \Lsym{]}  \Lnt{t_{{\mathrm{2}}}}  \rightsquigarrow   \mathsf{discard}\, \Lnt{t} \,\mathsf{in}\, \Lsym{[}  \Lnt{t_{{\mathrm{1}}}}  \Lsym{/}  \Lmv{x}  \Lsym{]}  \Lnt{t_{{\mathrm{2}}}}  }{%
{\LdruleRXXNatDName}{}%
}}

\newcommand{\LdruleRXXNatCName}[0]{\Ldrulename{R\_NatC}}
\newcommand{\LdruleRXXNatC}[1]{\Ldrule[#1]{%
}{
 \Lsym{[}   \mathsf{copy}\, \Lnt{t} \,\mathsf{as}\, \Lmv{x} , \Lmv{y} \,\mathsf{in}\, \Lnt{t_{{\mathrm{1}}}}   \Lsym{/}  \Lmv{x}  \Lsym{]}  \Lnt{t_{{\mathrm{2}}}}  \rightsquigarrow   \mathsf{copy}\, \Lnt{t} \,\mathsf{as}\, \Lmv{x} , \Lmv{y} \,\mathsf{in}\, \Lsym{[}  \Lnt{t_{{\mathrm{1}}}}  \Lsym{/}  \Lmv{x}  \Lsym{]}  \Lnt{t_{{\mathrm{2}}}}  }{%
{\LdruleRXXNatCName}{}%
}}

\newcommand{\LdruleRXXBetaEDRName}[0]{\Ldrulename{R\_BetaEDR}}
\newcommand{\LdruleRXXBetaEDR}[1]{\Ldrule[#1]{%
}{
  \mathsf{derelict}_\kappa\,  (   \mathsf{promote}_\kappa\,  \overrightarrow{ \Lnt{t} }  \,\mathsf{for}\,  \overrightarrow{ \Lmv{x} }  \,\mathsf{in}\, \Lnt{t_{{\mathrm{1}}}}   )    \rightsquigarrow  \Lsym{[}   \overrightarrow{ \Lnt{t} }   \Lsym{/}   \overrightarrow{ \Lmv{x} }   \Lsym{]}  \Lnt{t_{{\mathrm{1}}}} }{%
{\LdruleRXXBetaEDRName}{}%
}}

\newcommand{\LdruleRXXNatElName}[0]{\Ldrulename{R\_NatEl}}
\newcommand{\LdruleRXXNatEl}[1]{\Ldrule[#1]{%
}{
 \Lsym{[}   \mathsf{exchange_l}\, \Lnt{t_{{\mathrm{1}}}} , \Lnt{t_{{\mathrm{2}}}} \,\mathsf{with}\, \Lmv{x} , \Lmv{y} \,\mathsf{in}\, \Lnt{t_{{\mathrm{3}}}}   \Lsym{/}  \Lmv{z}  \Lsym{]}  \Lnt{t_{{\mathrm{4}}}}  \rightsquigarrow   \mathsf{exchange_l}\, \Lnt{t_{{\mathrm{1}}}} , \Lnt{t_{{\mathrm{2}}}} \,\mathsf{with}\, \Lmv{x} , \Lmv{y} \,\mathsf{in}\, \Lsym{[}  \Lnt{t_{{\mathrm{3}}}}  \Lsym{/}  \Lmv{z}  \Lsym{]}  \Lnt{t_{{\mathrm{4}}}}  }{%
{\LdruleRXXNatElName}{}%
}}

\newcommand{\LdruleRXXNatErName}[0]{\Ldrulename{R\_NatEr}}
\newcommand{\LdruleRXXNatEr}[1]{\Ldrule[#1]{%
}{
 \Lsym{[}   \mathsf{exchange_r}\, \Lnt{t_{{\mathrm{1}}}} , \Lnt{t_{{\mathrm{2}}}} \,\mathsf{with}\, \Lmv{x} , \Lmv{y} \,\mathsf{in}\, \Lnt{t_{{\mathrm{3}}}}   \Lsym{/}  \Lmv{z}  \Lsym{]}  \Lnt{t_{{\mathrm{4}}}}  \rightsquigarrow   \mathsf{exchange_r}\, \Lnt{t_{{\mathrm{1}}}} , \Lnt{t_{{\mathrm{2}}}} \,\mathsf{with}\, \Lmv{x} , \Lmv{y} \,\mathsf{in}\, \Lsym{[}  \Lnt{t_{{\mathrm{3}}}}  \Lsym{/}  \Lmv{z}  \Lsym{]}  \Lnt{t_{{\mathrm{4}}}}  }{%
{\LdruleRXXNatErName}{}%
}}

\renewcommand{\Ldrule}[4][]{{\displaystyle\frac{\begin{array}{l}#2\end{array}}{#3}\,\,\Ldrulename{{\scriptsize #4}}}}

\usepackage{hyperref}

\begin{document}

\title{Dialectica Categories for the Lambek Calculus}

\author{Valeria de Paiva\inst{1} and Harley Eades III\inst{2}}

\institute{Nuance Communications, Sunnyvale, CA, \email{valeria.depaiva@gmail.com}
  \and Computer Science, Augusta University, Augusta, GA, \email{harley.eades@gmail.com}}

\maketitle

\begin{abstract}
  We revisit the old work of de Paiva on the models of the Lambek
  Calculus in dialectica models making sure that the syntactic
  details that were sketchy on the first version got completed and
  verified.  We extend the Lambek Calculus with a $\kappa$ modality,
  inspired by Yetter's work, which makes the calculus
  commutative. Then we add the of-course modality $!$, as Girard did,
  to re-introduce weakening and contraction for all formulas and get
  back the full power of intuitionistic and classical logic. We also
  present the categorical semantics, proved sound and
  complete. Finally we show the traditional properties of type
  systems, like subject reduction, the Church-Rosser theorem and
  normalization for the calculi of extended modalities, which we did
  not have before.
\end{abstract}

\section*{Introduction}
Lambek introduced his homonymous calculus (originally called the
`Syntactic Calculus') for proposed applications in Linguistics.
However the calculus got much of its cult following and reputation by
being a convenient, well-behaved prototype of a Gentzen sequent
calculus without any structural rules.

This note recalls a Dialectica model of the Lambek Calculus presented
by the first author in the Amsterdam Colloquium in 1991
\cite{depaiva1991}. Here, like then, we approach the Lambek Calculus
from the perspective of Linear Logic, so we are interested in the
basic sequent calculus with no structural rules, except associativity
of tensors. In that earlier work we took for granted the syntax of the
calculus and only discussed the exciting possibilities of categorical
models of linear-logic-like systems.  Many years later we find that
the work on models is still interesting and novel, and that it might
inform some of the most recent work on the relationship between
categorial grammars and notions of distributional semantics
\cite{coecke2013}.


Moreover, using the Agda proof assistant \cite{bove2009} we verify the
correctness of our categorical model
(Section~\ref{sec:dialectica_lambek_spaces}), and we add the type
theoretical (Section~\ref{sec:typed_lambek_calculi}) notions that were
left undiscussed in the Amsterdam Colloquium presentation.  All of the
syntax in this paper was checked using \texttt{Ott}
\cite{Sewell:2010}.  The goal is to show that our work can shed new
light on some of the issues that remained open.  Mostly we wanted to
check the correctness of the semantic proposals put forward since
Szabo's seminal book \cite{szabo1978} and, for future work, on the
applicability and fit of the original systems to their intended uses.

\textbf{Overview.} The Syntactic Calculus was first introduced by
Joachim Lambek in 1958 \cite{Lambek1958}. Since then the rechristened
Lambek Calculus has had as its main motivation providing an
explanation of the mathematics of sentence structure, starting from
the author's algebraic intuitions. The Lambek Calculus is the core of
logical Categorial Grammar.  The first use of the term “categorial
grammar” seems to be in the title of Bar-Hillel, Gaifman and Shamir
(1960), but categorial grammar began with Ajdukiewicz (1935) quite a
few years earlier. After a period of ostracism, around 1980 the Lambek
Calculus was taken up by logicians interested in Computational
Linguistics, especially the ones interested in Categorial
Grammars. 

The work on Categorial Grammar was given a serious impulse by the
advent of Girard's Linear Logic at the end of the 1980s.  Girard
\cite{Girard:1987} showed that there is a full embedding, preserving
proofs, of Intuitionistic Logic into Linear Logic with a modality
``!". This meant that Linear Logic, while paying attention to
resources, could always code up faithfully Classical Logic and hence
one could, as Girard put it, `have one's cake and eat it', paying
attention to resources, if desired, but always forgetting this
accounting, if preferred. This meant also that several new systems of
resource logics were conceived and developed and these refined
resource logics were applied to several areas of Computer Science.

In Computational Linguistics, the Lambek Calculus has seen a
significant number of works written about it, apart from a number of
monographs that deal with logical and linguistic aspects of the
generalized type-logical approach.  For a shorter introduction, see
Moortgat's entry on the Stanford Encyclopedia of Philosophy on Type
Logical Grammar \cite{MoortgatSEP}.  Type Logical Grammar situates the
type-logical approach within the framework of Montague's Universal
Grammar and presents detailed linguistic analyses for a substantive
fragment of syntactic and semantic phenomena in the grammar of
English.  Type Logical Semantics offers a general introduction to
natural language semantics studied from a type-logical perspective.


This meant that a series of systems, implemented or not, were devised
that used the Lambek Calculus or its variants as their basis. These
systems can be as expressive as Intuitionistic Logic and the claim is
that they are more precise i.e. they make finer distinctions. Some of
the landscape of calculi can be depicted as follows:
\begin{center}
  \includegraphics[scale=0.30]{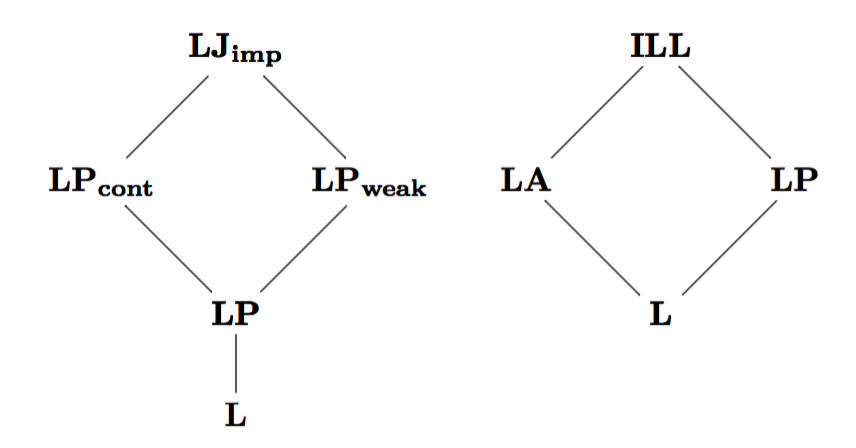}
\end{center}




From the beginning it was clear that the Lambek Calculus is the
multiplicative fragment of non-commutative Intuitionistic Linear
Logic. In the diagrams $\bf L$ stands for the Lambek Calculus, as
expounded in \cite{Lambek1958} but with the unit $I$ added for the
tensor connective (there was a certain amount of dispute on that, as
the original system did not introduce the constant corresponding to
the nullary case of the tensor product, here written as $I$). The
system $\bf{LP}$ is the Lambek Calculus with permutation, sometimes
called the van Benthem calculus. We are calling $\bf{LA}$ the Lambek
Calculus with additives, the more usual algebraic connectives
corresponding to meet and join. Hence adding both permutation and
additives to the Lambek Calculus we get to intuitonistic linear
logic. On the other diagram to the Lambek Calculus with permutation we
add either weakening ($\bf{LP_{weak}}$) or contraction
$\bf{LP_{cont}}$) or both to get the implicational fragment of
Intuitionsitic Propositional Logic.

The Lambek Calculus also has the potential for many applications in
other areas of computer science, such as, modeling processes.  Linear
Logic has been at the forefront of the study of process calculi for
many years \cite{HONDA20102223,Pratt:1997,ABRAMSKY19945}. We can think
of the commutative tensor product of linear logic as a parallel
operator.  For example, given a process $\Lnt{A}$ and a process $\Lnt{B}$,
then we can form the process $ \Lnt{A}  \otimes  \Lnt{B} $ which runs both processes
in parallel.  If we remove commutativity from the tensor product we
obtain a sequential composition instead of parallel composition.  That
is, the process $\Lnt{A} \rhd \Lnt{B}$ first runs process $\Lnt{A}$ and then
process $\Lnt{B}$ in that order.  Paraphrasing Vaughan Pratt, ``The
sequential composition operation has no evident counterpart in type
theory'' see page 11 of \cite{Pratt:1997}.  We believe that the Lambek
Calculus will lead to filling this hole, and the results of this paper
as a means of obtaining a theory with both a parallel operator and a
sequential composition operator.  This work thus has a potential to
impact research in programming languages and computer security where
both linear logic and sequential composition play important roles.

There are several interesting questions, considered for Linear Logic,
that could also be asked of the Lambek Calculus or its extensions.
One of them, posed by Morrill et al is whether we can extend the
Lambek Calculus with a modality that does for the structural rule of
\textit{(exchange)} what the modality \textit{of course} `!' does for
the rules of \textit{(weakening)} and \textit{(contraction)}.  A
preliminary proposal, which answers this question affirmatively, is
set forward in this paper. The answer was provided in semantical terms
in the first version of this work. Here we provide also the more
syntactic description of these modalities. Building up from work of
Ciabattoni, Galatos and Terui in \cite{Ciabattoni2012} and others that
describe how to transform systems of axioms into cut-free sequent
rules, we aim to refine the algebraization of proof theory.




\section{Related Work}
\label{sec:related_work}

Lamarche and Retor\'e \cite{Lamarche:1996} give an overview of proof
nets for the Lambek Calculus where they discuss the addition of
various exchange rules to the calculus.  For example, the following
permutation and cycle rules:
\[
\begin{array}{lll}
  $$\mprset{flushleft}
  \inferrule* [right=perm] {
    A_1, A_2, \Gamma \vdash B
  }{A_2, A_1, \Gamma \vdash B}
  & \quad &
  $$\mprset{flushleft}
  \inferrule* [right=cycl] {
    A_1, A_2, \Gamma \vdash B
  }{A_1, \Gamma, A_2, \Gamma \vdash B}
\end{array}
\]
Taken in isolation each rule does not imply general exchange, but
taken together they do.  Thus, it is possible to take a restricted
notion of exchange to the Lambek Calculus without taking general
exchange. However, applications where one needs a completely
non-commutative tensor product and a commutative tensor product cannot
be modeled in the Lambek Calculus with these rules.

Polakow and Pfenning~\cite{Polakow:2001} combine intuitionistic,
commutative linear, and non-commutative linear logic into a system
called Ordered Linear Logic (OLL).  Polakow and Pfenning then extend
OLL into two new systems: a term assignment of OLL called OLLi, and a
logical framework in the tradition of LF called OLF.

OLL's sequent is of the form $\Gamma, \Delta, \Omega \vdash A$ where
$\Gamma$ is a multiset of intuitionistic assumptions, $\Delta$ is a
multiset of commutative linear assumptions, and $\Omega$ is a list of
non-commutative linear assumptions.  Furthermore, OLL contains logical
connectives from each logic.  For example, there are two different
tensor products and three different implications.  Thus, the systems
developed here are a simplification of OLL showing how to get all of
these systems using modalities.

Greco and Palmigiano~\cite{2016arXiv161104181G} give a type of logic
called a proper display logic for the Lambek Calculus with
exponentials. However, they decompose the linear exponential into an
adjunction in the spirit Benton's LNL models.  In this paper we
concentrate on sequent calculi rather than display logic.

\section{The Lambek Calculus}

The Lambek Calculus, formerly the Syntactic Calculus $\sf L$, due to
J.  Lambek \cite{Lambek1958}, was created to capture the logical
structure of sentences.  Lambek introduced what we think of as a
substructural logic with an operator denoting concatenation,
$ \Lnt{A}  \otimes  \Lnt{B} $, and two implications relating the order of phrases,
$ \Lnt{A}  \leftharpoonup  \Lnt{B} $ and $ \Lnt{A}  \rightharpoonup  \Lnt{B} $.  The first implication corresponds to a
phrase of type $\Lnt{A}$ when followed by a phrase of type $\Lnt{B}$, and
the second is a phrase of type $\Lnt{B}$ when preceded by a phrase of
type $\Lnt{A}$.

The Lambek Calculus, defined in Figure~\ref{fig:L}, can be presented
as a non-commutative intuitionistic multiplicative linear logic.
Contexts are sequences of formulas, and we denote mapping the
modalities over an arbitrary context by $ !  \Gamma $ and $ \kappa  \Gamma $.
\begin{figure}
  \begin{mdframed}
    \begin{mathpar}
      \Ldruleax{} \and
      \LdruleUr{} \and      
    \Ldrulecut{} \and
    \LdruleUl{} \and
    \LdruleTl{} \and
    \LdruleTr{} \and
    \LdruleIRl{} \and
    \LdruleILl{} \and
    \LdruleIRr{} \and
    \LdruleILr{}     
  \end{mathpar}
  \end{mdframed}
    
  \caption{The Lambek Calculus: L}
  \label{fig:L}
\end{figure}

Because the operator $ \Lnt{A}  \otimes  \Lnt{B} $ denotes the type of concatenations
the types $ \Lnt{A}  \otimes  \Lnt{B} $ and $ \Lnt{B}  \otimes  \Lnt{A} $ are not equivalent, and
hence, \textsf{L} is non-commutative which explains why implication
must be broken up into two operators $ \Lnt{A}  \leftharpoonup  \Lnt{B} $ and $ \Lnt{A}  \rightharpoonup  \Lnt{B} $.
In the following subsections we give two extensions of L: one with the
well-known modality of-course of linear logic which adds weakening and
contraction, and a second with a new modality adding exchange.

\section{
Extensions to the Lambek Calculus}
\label{subsec:the_lambek_calculus_with_the_weakening_and_contraction_modality}

The linear modality, $ !  \Lnt{A} $, read ``of-course $\Lnt{A}$'' was first
proposed by Girard \cite{Girard:1987} as a means of encoding
non-linear logic in both classical and intuitionistic forms in linear
logic.  For example, non-linear implication $ \Lnt{A}  \rightharpoonup  \Lnt{B} $ is usually
encoded into linear logic by $ !  \Lnt{A}  \multimap \Lnt{B}$. Since we have
based L on non-commutative intuitionistic linear logic it is
straightforward to add the of-course modality to L.  The rules for the
of-course modality are defined by the following rules:
\[
\small
\begin{array}{ccccccccccc}
  \LdruleC{} & \quad & \LdruleW{} & \quad & 
  \LdruleBr{} & \quad & \LdruleBl{}
\end{array}
\]
The rules $\Ldrulename{C}$ and $\Ldrulename{W}$ add contraction and
weakening to L in a controlled way.  Then the other two rules allow
for linear formulas to be injected into the modality; and essentially
correspond to the rules for necessitation found in S4
\cite{bierman2000}.  Thus, under the of-course modality the logic
becomes non-linear. We will see in
Section~\ref{sec:dialectica_lambek_spaces} that these rules define a
comonad.  We call the extension of L with the of-course modality
$\LB$.

As we remarked above, one leading question of the Lambek Calculus is:
can exchange be added in a similar way to weakening and contraction?
That is, can we add a new modality that adds the exchange rule to L in
a controlled way?  The answer to this question is positive, and the
rules for this new modality are as follows:
\[
\small
\begin{array}{ccccccccccccccccccccc}  
  \LdruleEr{} & & \LdruleEl{} & & \LdruleEOne{} & & \LdruleETwo{} 
\end{array}
\]
The first two rules are similar to of-course, but the last two add
exchange to formulas under the $\kappa$-modality.  We call L with the
exchange modality $\LE$.  Thus, unlike intuitionistic linear logic
where any two formulas can be exchanged, $\LE$ restricts exchange to
only formulas under the exchange modality.  Just like of-course
the exchange modality is modeled categorically as a comonad; see
Section~\ref{sec:dialectica_lambek_spaces}.

\section{Categorical Models}

We now turn to the categorical models, ones where one considers
different proofs of the same theorem. Since the Lambek Calculus itself
came from its categorical models, biclosed monoidal categories, there
is no shortage of these models. However, Girard's insight of relating
logical systems via modalities should also be considered in this
context.

Lambek's work on monoidal biclosed categories happened almost three
decades before Girard introduced Linear Logic, hence there were no
modalities or exponentials in Lambek's setting. The categorical
modelling of the modalities (of-course! and why-not?) was the
difficult issue with Linear Logic.
This is where there are design decisions to be made. 


\subsection{Dialectica Lambek Spaces}
\label{sec:dialectica_lambek_spaces}

A sound and complete categorical model of  the Lambek Calculi
 can be given using a modification of de Paiva's
dialectica categories \cite{depaiva1990}.  Dialectica
categories arose from de Paiva's thesis on a categorical model
of G\"odel's Dialectica interpretation, hence the name.  Dialectica
categories were one of the first sound categorical models of
intuitionistic linear logic, with  linear modalities.  We show in
this section that they can be adapted to become a sound and complete model for the Lambek Calculus, with both the
exchange and of-course modalities. We call this model  \textit{dialectica Lambek spaces}.

Due to the complexities of working
with dialectica categories we have formally verified\footnote{The complete formalization can be
  found online at
  \url{https://github.com/heades/dialectica-spaces/blob/Lambek/NCDialSets.agda}.} this section in
the proof assistant Agda~\cite{bove2009}.
Dialectica categories arise as constructions over a given monoidal
category.  Suppose $\cat{C}$ is such a category.  Then in complete
generality the objects of the dialectica category over $\cat{C}$ are
triples $(U, X, \alpha)$ where $U$ and $X$ are objects of $\cat{C}$,
and $\alpha : A \mto U \otimes X$ is a subobject of the tensor product
in $\cat{C}$ of $U$ and $X$.  Thus, we can think of $\alpha$ as a
relation over $U \otimes X$.  If we specialize the category $\cat{C}$
to the category of sets and functions, $\Set$, then we obtain what is
called a dialectica space. Dialectica spaces are a useful model of
full intuitionistic linear logic \cite{Hyland:1993}.

Morphisms between objects $(U, X, \alpha)$ and $(V, Y, \beta)$ are
pairs $(f, F)$ where $f : U \mto V$ and $F : Y \mto X$ are morphisms
of $\cat{C}$ such that the pullback condition 
$(U \otimes F)^{-1}(\alpha) \leq (f \otimes Y)^{-1}(\beta)$ holds.
In dialectica spaces this condition becomes
$\forall u \in U.\forall y \in Y. \alpha(u , F(y)) \leq \beta(f(u), y)$.
The latter reveals that we can think of the condition on morphisms as
a weak form of an adjoint condition.  Finally, through some nontrivial reasoning
on this structure we can show that this is indeed a category; for the
details see the formal development.  Dialectica categories are related
to the Chu construction \cite{depaiva2007} and to Lafont and Streicher's category of games $\text{GAME}_{\kappa}$~\cite{lafont1991}.

To some extent the underlying category $\cat{C}$ controls the kind of
structure we can expect in the dialectica category over $\cat{C}$. However,
de Paiva  showed \cite{depaiva2007} that by changing the relations used in
the objects and the order used in the  `adjoint condition' (which also
controls the type of structure) we can obtain a non-symmetric tensor in the dialectica category, if the structure of the underlying category and the structure of the underlying relations are compatible. 
She also showed that one can
abstract the notion of relation out as a parameter in the dialectica
construction, so long as this has enough structure, i.e. so long as you have an algebra (that she called a lineale) to evaluate the relations at. 
We denote this construction by $\mathsf{Dial}_{L}(\cat{C})$ where $L$ is the lineale controlling the relations coming from the monoidal category $\cat{C}$. For example, $\mathsf{Dial}_{\mathsf{2}}(\Set)$ is the
category of usual dialectica spaces of sets over the Heyting (or Boolean) algebra $2$. 

This way we can see dialectica categories as a framework of categorical models of various logics, varying the underlying category $\cat{C}$ as well as the underlying lineale or algebra of relations $L$. Depending on which category we start with and which structure we use for the relations in the construction we will obtain different models for different logics.

The underlying category we will choose here is the category $\Set$,
but the structure we will define our relations over will be a biclosed
poset, defined in the next definition.
\begin{definition}
  \label{def:biclosed-poset}
  Suppose $(M, \leq, \circ, e)$ is an ordered non-commutative monoid.
  If there exists a largest $x \in M$ such that $a \circ x \leq b$ for
  any $a, b \in M$, then we denote $x$ by $a \lto b$ and called it
  the \textbf{left-pseudocomplement} of $a$ w.r.t $b$.  Additionally,
  if there exists a largest $x \in M$ such that $x \circ a \leq b$ for
  any $a, b \in M$, then we denote $x$ by $b \rto a$ and called it
  the \textbf{right-pseudocomplement} of $a$ w.r.t $b$.

  A \textbf{biclosed poset}, $(M, \leq, \circ, e, \lto, \rto)$, is an
  ordered non-commutative monoid, $(M, \leq, \circ, e)$, such that $a
  \lto b$ and $b \rto a$ exist for any $a,b \in M$.
\end{definition}
\noindent
Now using the previous definition we define dialectica Lambek spaces.
\begin{definition}
  \label{def:dialectica-lambek-spaces}
  Suppose $(M, \leq, \circ, e, \lto, \rto)$ is a biclosed poset. Then
  we define the category of \textbf{dialectica Lambek spaces},
  $\mathsf{Dial}_M(\Set)$, as follows:
  \begin{itemize}
  \item[-] objects, or dialectica Lambek spaces, are triples $(U, X,
    \alpha)$ where $U$ and $X$ are sets, and $\alpha : U \times X \mto
    M$ is a generalized relation over $M$, and

  \item[-] maps that are pairs $(f, F) : (U , X, \alpha) \mto (V , Y ,
    \beta)$ where $f : U \mto V$, and $F : Y \mto X$ are functions
    such that the weak adjointness condition
    $\forall u \in U.\forall y \in Y. \alpha(u , F(y)) \leq \beta(f(u), y)$
    holds.
  \end{itemize}
\end{definition}
Notice that the biclosed poset is used here as the target of the
relations in objects, but also as providing the order  relation in the weak adjoint condition on morphisms.  This will allow the structure of the biclosed
poset to lift up into $\Dial{M}{\Set}$.

We will show that $\Dial{M}{\Set}$ is a model of the Lambek Calculus with
modalities.  First, we show that it is a model of the Lambek Calculus
without modalities.  Thus, we must show that $\Dial{M}{\Set}$ is
monoidal biclosed.

\begin{definition}
  \label{def:dial-monoidal-structure}
  Suppose $(U, X, \alpha)$ and $(V, Y, \beta)$ are two objects of
  $\Dial{M}{\Set}$. Then their tensor product is defined as follows:
  \[ \small
  (U, X, \alpha) \otimes (V, Y, \beta) = (U \times V, (V \to X) \times (U \to Y), \alpha \otimes \beta)
  \]
  where $- \to -$ is the function space from $\Set$, and $(\alpha
  \otimes \beta)((u, v), (f, g)) = \alpha(u, f(v)) \circ \beta(g(u), v)$.
\end{definition}

\noindent
The identity of the tensor product just defined is $I = (\mathbbm{1},
\mathbbm{1}, e)$, where $\mathbbm{1}$ is the terminal object in
$\Set$, and $e$ is the unit of the biclosed poset.  It is
straightforward to show that the tensor product is functorial, one can
define the left and right unitors, and the associator for tensor; see
the formalization for the definitions.  In addition, all of the usual
monoidal diagrams hold \cite{depaiva1990}.  Take note of the fact that
this tensor product is indeed non-commutative, because the
non-commutative multiplication of the biclosed poset is used to define
the relation of the tensor product.

The tensor product has two right adjoints making $\Dial{M}{\Set}$
biclosed.
\begin{definition}
  \label{def:dial-is-biclosed}
  Suppose $(U, X, \alpha)$ and $(V, Y, \beta)$ are two objects of
  $\Dial{M}{\Set}$. Then two internal-homs can be defined as follows:
  \[ \small
  \begin{array}{lll}
    (U, X, \alpha) \lto (V, Y, \beta) = ((U \to V) \times (Y \to X), U \times Y, \alpha \lto \beta)\\
    (V, Y, \beta) \rto (U, X, \alpha) = ((U \to V) \times (Y \to X), U \times Y, \alpha \rto \beta)\\
  \end{array}
  \]
\end{definition}
These two definitions are functorial, where the first is contravariant
in the first argument and covariant in the second, but the second
internal-hom is covariant in the first argument and contravariant in
the second.  The relations in the previous two definitions prevent
these two from collapsing into the same object, because of the use of
the left and right pseudocomplement. It is straightforward to show
that the following bijections hold:
\[ \small
\begin{array}{lll}
  \Hom{}{A \otimes B}{C} \cong \Hom{}{B}{A \lto C} & \quad &  \Hom{}{A \otimes B}{C} \cong \Hom{}{A}{C \rto B}
\end{array}
\]
Therefore, $\Dial{M}{\Set}$ is biclosed, and we obtain the following
result.
\begin{theorem}
  \label{theorem:sound-lambek}
  $\Dial{M}{\Set}$ is a sound and complete model for the Lambek
  Calculus $L$ without modalities.
\end{theorem}

We now extend $\Dial{M}{\Set}$ with two modalities: the usual
modality, of-course, denoted $!A$, and the exchange modality denoted
$\kappa A$.  However, we must first extended biclosed posets to
include an exchange operation.

\begin{definition}
  \label{def:biclosed-exchange}
  A \textbf{biclosed poset with exchange} is a biclosed poset $(M,
  \leq, \circ, e, \lto, \rto)$ equipped with an unary operation
  $\kappa : M \to M$ satisfying the following:
  \[ \small
  \begin{array}{clr}
    \text{(Compatibility)} & a \leq b \text{ implies } \kappa a \leq \kappa b \text{ for all } a,b,c \in M\\
    \text{(Minimality)} & \kappa a \leq a \text{ for all } a \in M\\
    \text{(Duplication)} & \kappa a \leq \kappa\kappa a \text{ for all } a \in M\\
    \text{(Left Exchange)} & \kappa a \circ b \leq b \circ \kappa a \text{ for all } a, b \in M\\
    \text{(Right Exchange)} & a \circ \kappa b \leq \kappa b \circ a \text{ for all } a, b \in M\\
  \end{array}
  \]
\end{definition}
\noindent
Compatibility results in $\kappa : M \to M$ being a functor in the
biclosed poset, and the remainder of the axioms imply that $\kappa$ is
a comonad extending the biclosed poset with left and right exchange.

We can now define the two modalities in $\Dial{M}{\Set}$ where $M$ is
a biclosed poset with exchange; clearly we know $\Dial{M}{\Set}$ is
also a model of the Lambek Calculus without modalities by
Theorem~\ref{theorem:sound-lambek} because $M$ is a biclosed poset.
\begin{definition}
  \label{def:modalities-dial}
  Suppose $(U, X, \alpha)$ is an object of $\Dial{M}{\Set}$ where $M$
  is a biclosed poset with exchange. Then the \textbf{of-course} and
  \textbf{exchange} modalities can be defined as 
  $! (U, X, \alpha) = (U, U \to X^*, !\alpha)$ and
  $\kappa (U, X, \alpha) = (U, X, \kappa \alpha)$
  where $X^*$ is the free commutative monoid on $X$, $(!\alpha)(u, f)
  = \alpha(u, x_1) \circ \cdots \circ \alpha(u, x_i)$ for $f(u) =
  (x_1, \ldots, x_i)$, and $(\kappa \alpha)(u, x) = \kappa (\alpha(u,
  x))$.
\end{definition}
This definition highlights a fundamental difference between the two
modalities.  The definition of the exchange modality relies on an
extension of biclosed posets with essentially the exchange modality in
the category of posets.  However, the of-course modality is defined by
the structure already present in $\Dial{M}{\Set}$, specifically, the
structure of $\Set$.

Both of the modalities have the structure of a comonad.  That is,
there are monoidal natural transformations $\varepsilon_! : !A \mto
A$, $\varepsilon_\kappa : \kappa A \mto A$, $\delta_! : !A \mto !!A$,
and $\delta_\kappa : \kappa A \mto \kappa\kappa A$ which satisfy the
appropriate diagrams; see the formalization for the full
proofs. Furthermore, these comonads come equipped with arrows $e : !A
\mto I$, $d : !A \mto !A \otimes !A$, $\beta L : \kappa A \otimes B \mto B
\otimes \kappa A$, and $\beta R : A \otimes \kappa B \mto \kappa B
\otimes A$.  Thus, we arrive at the following result.

\begin{theorem}
  \label{theorem:sound-dial-exchange-!}
  Suppose $M$ is a biclosed poset with exchange.  Then
  $\Dial{M}{\Set}$ is a sound and complete model for the Lambek
  Calculi $L_!$, $L_\kappa$, and $L_{!\kappa}$.
\end{theorem}

\section{Type Theory for Lambek Systems}
\label{sec:typed_lambek_calculi}
In this section we introduce typed calculi for each of the logics
discussed so far.  Each type system is based on the term assignment
for  Intuitionistic Linear Logic introduced 
in \cite{benton1993}.
We show that they are all strongly normalizing
and confluent, but we do not give full detailed proofs of each of
these properties, because they are straightforward consequences of the
proofs of strong normalization and confluence for intuitionistic
linear logic.  In fact, we will reference Bierman's thesis often
within this section.  The reader may wish to review Section~3.5 on
page 88 of \cite{Bierman:1994}.

\subsection{The Typed Lambek Calculus: $\lambda\text{L}$}
\label{subsec:the_typed_lambek_calculus_lambda-l}

The first system we cover is the Lambek Calculus without
modalities. This system can be seen as the initial core of each of the
other systems we introduce below, and thus, we will simply extend the
results here to three other systems.  

The syntax for patterns, terms, and contexts are described by the
following grammar:
\vspace{-10px}
\[ \small
\begin{array}{cll}
  \text{(patterns)} & \Lnt{p} :=  -  \mid \Lmv{x} \mid  \mathsf{unit}  \mid  \Lnt{p_{{\mathrm{1}}}}  \otimes  \Lnt{p_{{\mathrm{2}}}} \\
  \text{(terms)}    & \Lnt{t} := \Lmv{x} \mid  \mathsf{unit}  \mid  \Lnt{t_{{\mathrm{1}}}}  \otimes  \Lnt{t_{{\mathrm{2}}}}  \mid  \lambda_l  \Lmv{x} : \Lnt{A} . \Lnt{t}  \mid  \lambda_r  \Lmv{x} : \Lnt{A} . \Lnt{t}  \mid
   \mathsf{app}_l\, \Lnt{t_{{\mathrm{1}}}} \, \Lnt{t_{{\mathrm{2}}}}  \mid \\ &  \mathsf{app}_r\, \Lnt{t_{{\mathrm{1}}}} \, \Lnt{t_{{\mathrm{2}}}}  \mid  \mathsf{let}\, \Lnt{t_{{\mathrm{1}}}} \,\mathsf{be}\, \Lnt{p} \,\mathsf{in}\, \Lnt{t_{{\mathrm{2}}}} \\
  \text{(contexts)} & \Gamma :=  \cdot  \mid \Lmv{x}  \Lsym{:}  \Lnt{A} \mid \Gamma_{{\mathrm{1}}}  \Lsym{,}  \Gamma_{{\mathrm{2}}}\\
\end{array}
\vspace{-10px}
\]
Contexts are sequences of pairs of free variables and types.  Patterns
are only used in the let-expression which is itself used to eliminate
logical connectives within the left rules of L.  All variables in the
pattern of a let-expression are bound.  The remainder of the terms are
straightforward.

The typing rules can be found in the in Figure~\ref{fig:typed-L} and
the reduction rules in Figure~\ref{fig:rewrite-L}. The typing rules
are as one might expect.  The reduction rules were extracted from the
cut-elimination procedure for L.

We denote the reflexive and transitive closure of the $ \redto $ by
$ \redto^* $.  We call a term with no $\beta$-redexes a normal form,
and we denote normal forms by $\Lnt{n}$.  In the interest of space we
omit the congruence rules from the definition of the reduction
relation; we will do this for each calculi introduced throughout this
section.  The other typed calculi we introduce below will be
extensions of $\lambda\text{L}$, thus, we do not reintroduce these
rules each time for readability.

\textbf{Strong normalization.}  It is well known that intuitionistic
linear logic (ILL) is strongly normalizing, for example, see Bierman's
thesis \cite{Bierman:1994} or Benton's beautiful embedding of ILL into
system F \cite{Benton:1995c}.  

It is fairly straightforward to define a reduction preserving
embedding of $\lambda\text{L}$ into ILL.  Intuitionistic linear logic
can be obtained from $\lambda\text{L}$ by replacing the rules
$\Ldrulename{T\_IRl}$, $\Ldrulename{T\_ILl}$, $\Ldrulename{T\_IRr}$,
and $\Ldrulename{T\_ILr}$ with the following two rules:
\[ \small
\begin{array}{cccccccc} 
  \LdruleTXXIl{} & \quad \LdruleTXXIr{}
\end{array}
\]
In addition, contexts are considered multisets, and hence, exchange is
handled implicitly. Then we can reuse the idea of Benton's embeddings
to show type preservation and type reduction.

At this point we define the following embeddings.
\begin{definition}
  \label{def:lambda-L_to_ILL}
  We embed types and terms of $\lambda\text{L}$ into ILL as follows:
  \begin{center}
    \begin{math} \small
      \begin{array}{lll}
        \begin{array}{lccc}
        \text{Types:}\\
        & \begin{array}{rll}
            \begin{array}{rll}
                I  ^{\mathsf{e} }  & = &  I \\
               \Lsym{(}   \Lnt{A}  \otimes  \Lnt{B}   \Lsym{)} ^{\mathsf{e} }  & = &   \Lnt{A} ^{\mathsf{e} }   \otimes    \Lnt{B} ^{\mathsf{e} }   \\
            \end{array}
            &
            \begin{array}{rll}
               \Lsym{(}   \Lnt{A}  \rightharpoonup  \Lnt{B}   \Lsym{)} ^{\mathsf{e} }  & = &   \Lnt{A} ^{\mathsf{e} }   \multimap    \Lnt{B} ^{\mathsf{e} }   \\
               \Lsym{(}   \Lnt{A}  \leftharpoonup  \Lnt{B}   \Lsym{)} ^{\mathsf{e} }  & = &   \Lnt{A} ^{\mathsf{e} }   \multimap    \Lnt{B} ^{\mathsf{e} }   \\                    
            \end{array}
          \end{array}                
      \end{array}
      \\
      \begin{array}{lll}
        \text{Terms:}\\
        & \begin{array}{rll}
             \Lmv{x} ^{\mathsf{e} }  & = & \Lmv{x}\\
              \mathsf{unit}  ^{\mathsf{e} }  & = &  \mathsf{unit} \\
              (   \Lnt{t_{{\mathrm{1}}}}  \otimes  \Lnt{t_{{\mathrm{2}}}}   )  ^{\mathsf{e} }  & = &   \Lnt{t_{{\mathrm{1}}}} ^{\mathsf{e} }   \otimes    \Lnt{t_{{\mathrm{2}}}} ^{\mathsf{e} }   \\
              (   \mathsf{let}\, \Lnt{t_{{\mathrm{1}}}} \,\mathsf{be}\, \Lnt{p} \,\mathsf{in}\, \Lnt{t_{{\mathrm{2}}}}   )  ^{\mathsf{e} }  & = &  \mathsf{let}\,  \Lnt{t_{{\mathrm{1}}}} ^{\mathsf{e} }  \,\mathsf{be}\, \Lnt{p} \,\mathsf{in}\,   \Lnt{t_{{\mathrm{2}}}} ^{\mathsf{e} }   \\
          \end{array}
        & \begin{array}{rll}
              (   \lambda_l  \Lmv{x} : \Lnt{A} . \Lnt{t}   )  ^{\mathsf{e} }  & = &  \lambda  \Lmv{x} : \Lnt{A} .   \Lnt{t} ^{\mathsf{e} }   \\
              (   \lambda_r  \Lmv{x} : \Lnt{A} . \Lnt{t}   )  ^{\mathsf{e} }  & = &  \lambda  \Lmv{x} : \Lnt{A} .   \Lnt{t} ^{\mathsf{e} }   \\
              (   \mathsf{app}_l\, \Lnt{t_{{\mathrm{1}}}} \, \Lnt{t_{{\mathrm{2}}}}   )  ^{\mathsf{e} }  & = &    \Lnt{t_{{\mathrm{1}}}} ^{\mathsf{e} }   \,   \Lnt{t_{{\mathrm{2}}}} ^{\mathsf{e} }   \\
              (   \mathsf{app}_r\, \Lnt{t_{{\mathrm{1}}}} \, \Lnt{t_{{\mathrm{2}}}}   )  ^{\mathsf{e} }  & = &    \Lnt{t_{{\mathrm{1}}}} ^{\mathsf{e} }   \,   \Lnt{t_{{\mathrm{2}}}} ^{\mathsf{e} }   \\
          \end{array}
      \end{array}
      \end{array}
    \end{math}
  \end{center}
  The previous embeddings can be extended to contexts in the
  straightforward way, and to sequents as follows:
  \[  \Lsym{(}   \Gamma  \vdash  \Lnt{t}  :  \Lnt{A}   \Lsym{)} ^{\mathsf{e} }  =   \Gamma ^{\mathsf{e} }   \vdash   \Lnt{t} ^{\mathsf{e} }   :    \Lnt{A} ^{\mathsf{e} }   \]
\end{definition}
\noindent
We can now prove strong normalization using the embedding preserves.
\begin{theorem}[Strong Normalization]
  \label{theorem:strong_normalization}
  \begin{itemize}
  \item If $ \Gamma  \vdash  \Lnt{t}  :  \Lnt{A} $ in $\lambda\text{L}$, then
    $  \Gamma ^{\mathsf{e} }   \vdash   \Lnt{t} ^{\mathsf{e} }   :    \Lnt{A} ^{\mathsf{e} }   $ in ILL.
  \item If $ \Lnt{t_{{\mathrm{1}}}}  \rightsquigarrow  \Lnt{t_{{\mathrm{2}}}} $ in $\lambda\text{L}$, then $  \Lnt{t_{{\mathrm{1}}}} ^{\mathsf{e} }   \rightsquigarrow    \Lnt{t_{{\mathrm{2}}}} ^{\mathsf{e} }   $
    in ILL.
  \item If $ \Gamma  \vdash  \Lnt{t}  :  \Lnt{A} $, then $\Lnt{t}$ is strongly normalizing.
  \end{itemize}
\end{theorem}
\begin{proof}
  The first two cases hold by straightforward induction on the form of the assumed
  typing or reduction derivation. They then imply the third.
\end{proof}
\begin{figure}[t]
  \small
  \begin{mdframed}
    \begin{mathpar}
      \LdruleTXXvar{} \and
      \LdruleTXXUr{} \and
      \LdruleTXXcut{} \and
      \LdruleTXXUl{} \and
      \LdruleTXXTl{} \and
      \LdruleTXXTr{} \and
      \LdruleTXXIRl{} \and
      \LdruleTXXILl{} \and
      \LdruleTXXIRr{} \and
      \LdruleTXXILr{} 
    \end{mathpar}
  \end{mdframed}
  \caption{Typing Rules for the Typed Lambek Calculus: $\lambda\text{L}$}
  \label{fig:typed-L}
\end{figure}
\begin{figure}[h]
  \small
  \begin{mdframed}
    \begin{mathpar}
      \LdruleRXXBetal{} \and
      \LdruleRXXBetar{} \and
      \LdruleRXXBetaU{} \and
      \LdruleRXXBetaTOne{} \and
      \LdruleRXXBetaTTwo{} \and
      \LdruleRXXNatU{} \and
      \LdruleRXXNatT{} \and
      \LdruleRXXLetU{} \and
    \end{mathpar}
  \end{mdframed}
  \caption{Rewriting Rules for the Lambek Calculus: $\lambda\text{L}$}
  \label{fig:rewrite-L}
\end{figure}
\noindent

\textbf{Confluence.} The Church-Rosser property is well known to hold
for ILL modulo commuting conversions, for example, see Theorem~19 of
\cite{Bierman:1994} on page 96.  Since $\lambda\text{L}$ is
essentially a subsystem of ILL, it is straightforward, albeit
lengthly, to simply redefine Bierman's candidates
and carry out a similar proof as Bierman's (Theorem~19 on page 96 of ibid.).  
\begin{theorem}[Confluence]
  \label{theorem:confluence}
  The reduction relation, $ \redto $, modulo the commuting conversions
  is confluent.
\end{theorem}

\subsection{The Typed Lambek Calculus: $\lambda\text{L}_!$}
\label{subsec:the_typed_lambek_calculus:lambda-l!}
The calculus we introduce in this section is an extension of
$\lambda\text{L}$ with the of-course modality $ !  \Lnt{A} $.  This
extension follows from ILL exactly.  The syntax of types and terms of
$\lambda\text{L}$ are extended as follows:
\[ \small
\begin{array}{cllllll}
  \text{(types)}    & \Lnt{A} & := & \cdots \mid  !  \Lnt{A} \\
  \text{(terms)}    & \Lnt{t} & := & \cdots \mid  \mathsf{copy}\, \Lnt{t'} \,\mathsf{as}\, \Lnt{t_{{\mathrm{1}}}} , \Lnt{t_{{\mathrm{2}}}} \,\mathsf{in}\, \Lnt{t}  \mid  \mathsf{discard}\, \Lnt{t'} \,\mathsf{in}\, \Lnt{t} 
  \mid  \mathsf{promote}_!\, \Lnt{t'} \,\mathsf{for}\, \Lnt{t''} \,\mathsf{in}\, \Lnt{t}  \mid \\ & & &  \mathsf{derelict}_!\, \Lnt{t} \\
\end{array}
\]
The new type and terms are what one might expect, and are the
traditional syntax used for the of-course modality.  We add the
following typing rules to $\lambda\text{L}$:
{\small
\begin{mathpar} 
      \LdruleTXXC{} \and
      \LdruleTXXW{} \and
      \LdruleTXXBr{} \and
      \LdruleTXXBl{} 
\end{mathpar}
}
\begin{figure}
  \footnotesize
  \begin{mdframed}
    \begin{mathpar}
      \LdruleRXXBetaDR{} \and
      \LdruleRXXBetaDI{} \and
      \LdruleRXXBetaC{} \and
      \LdruleRXXNatD{} \and
      \LdruleRXXNatC{} 
    \end{mathpar}
  \end{mdframed}
  \caption{Rewriting Rules for The Typed Lambek Calculus: $\lambda\text{L}_!$}
  \label{fig:rewrite-LB}
\end{figure}
\noindent
Finally, the reduction rules can be found in
Figure~\ref{fig:rewrite-LB}.  The equality used in the
$\Ldrulename{R\_BetaC}$ rule is definitional, meaning, that the rule
simply gives the terms on the right side of the equation the name on
the left side, and nothing more.  This makes the rule easier to read.

\textbf{Strong normalization.} Showing strong normalization for
$\lambda\text{L}_!$ easily follows by a straightforward extension of
the embedding we gave for $\lambda\text{L}$.
\begin{definition}
  \label{def:embed-lambda-L!-in-ILL}
  The following is an extension of the embedding of $\lambda\text{L}$
  into ILL resulting in an embedding of types and terms of
  $\lambda\text{L}_!$ into ILL. First, we define $ \Lsym{(}   !  \Lnt{A}   \Lsym{)} ^{\mathsf{e} }  =  !    \Lnt{A} ^{\mathsf{e} }   $, then the following defines the embedding of terms:
  \begin{center} \small
    \begin{math}
      \begin{array}{rlllllllllllllllllll}        
          (   \mathsf{copy}\, \Lnt{t'} \,\mathsf{as}\, \Lnt{t_{{\mathrm{1}}}} , \Lnt{t_{{\mathrm{2}}}} \,\mathsf{in}\, \Lnt{t}   )  ^{\mathsf{e} }  & = &  \mathsf{copy}\,  \Lnt{t'} ^{\mathsf{e} }  \,\mathsf{as}\,  \Lnt{t_{{\mathrm{1}}}} ^{\mathsf{e} }  ,  \Lnt{t_{{\mathrm{2}}}} ^{\mathsf{e} }  \,\mathsf{in}\,   \Lnt{t} ^{\mathsf{e} }   \\
          (   \mathsf{discard}\, \Lnt{t'} \,\mathsf{in}\, \Lnt{t}   )  ^{\mathsf{e} }  & = &  \mathsf{discard}\,  \Lnt{t'} ^{\mathsf{e} }  \,\mathsf{in}\,   \Lnt{t} ^{\mathsf{e} }   \\
          (   \mathsf{promote}_!\, \Lnt{t'} \,\mathsf{for}\, \Lnt{t''} \,\mathsf{in}\, \Lnt{t}   )  ^{\mathsf{e} }  & = &  \mathsf{promote}_!\,  \Lnt{t'} ^{\mathsf{e} }  \,\mathsf{for}\,  \Lnt{t''} ^{\mathsf{e} }  \,\mathsf{in}\,   \Lnt{t} ^{\mathsf{e} }   \\
          (   \mathsf{derelict}_!\, \Lnt{t}   )  ^{\mathsf{e} }  & = &  \mathsf{derelict}_!\,   \Lnt{t} ^{\mathsf{e} }   \\
      \end{array}
    \end{math}
  \end{center}
\end{definition}

\noindent
Just as before this embedding is type preserving and reduction preserving.
\begin{theorem}[Type and Reduction Preserving Embedding]
  \label{theorem:strong_normalization_lambdaL!}
  \begin{itemize}
  \item  If $ \Gamma  \vdash  \Lnt{t}  :  \Lnt{A} $ in $\lambda\text{L}_!$, then
  $  \Gamma ^{\mathsf{e} }   \vdash   \Lnt{t} ^{\mathsf{e} }   :    \Lnt{A} ^{\mathsf{e} }   $ in ILL.

  \item If $ \Lnt{t_{{\mathrm{1}}}}  \rightsquigarrow  \Lnt{t_{{\mathrm{2}}}} $ in $\lambda\text{L}_!$, then $  \Lnt{t_{{\mathrm{1}}}} ^{\mathsf{e} }   \rightsquigarrow    \Lnt{t_{{\mathrm{2}}}} ^{\mathsf{e} }   $ in ILL.
  \item If $ \Gamma  \vdash  \Lnt{t}  :  \Lnt{A} $, then $\Lnt{t}$ is strongly normalizing.  
  \end{itemize}
\end{theorem}
\begin{proof}
  The first two cases hold by straightforward induction on the form of the assumed
  typing or reduction derivation.  They then imply the third.
\end{proof}
\noindent

\textbf{Confluence.} The Church-Rosser property also holds for
$\lambda\text{L}_!$, and can be shown by straightforwardly applying a
slightly modified version of Bierman's proof \cite{Bierman:1994} just
as we did for $\lambda\text{L}$.  Thus, we have the following:
\begin{theorem}[Confluence]
  \label{theorem:confluence}
  The reduction relation, $ \redto $, modulo the commuting conversions
  is confluent.
\end{theorem}

\subsection{The Typed Lambek Calculus: $\lambda\text{L}_\kappa$}
\label{subsec:the_typed_lambek_calculus:lambda-l-kappa}

The next calculus we introduce is also an extension of
$\lambda\text{L}$ with a modality that adds exchange to
$\lambda\text{L}_\kappa$ denoted $ \kappa  \Lnt{A} $.  It is perhaps the most
novel of the calculi we have introduced.

The syntax of types and terms of $\lambda\text{L}$ are extended as
follows:
\[ \small
\begin{array}{cllllll}
  \text{(types)}    & \Lnt{A} & := & \cdots \mid  \kappa  \Lnt{A} \\
  \text{(terms)}    & \Lnt{t} & := & \cdots \mid  \mathsf{exchange_l}\, \Lnt{t_{{\mathrm{1}}}} , \Lnt{t_{{\mathrm{2}}}} \,\mathsf{with}\, \Lmv{x} , \Lmv{y} \,\mathsf{in}\, \Lnt{t_{{\mathrm{3}}}}  \mid  \mathsf{exchange_r}\, \Lnt{t_{{\mathrm{1}}}} , \Lnt{t_{{\mathrm{2}}}} \,\mathsf{with}\, \Lmv{x} , \Lmv{y} \,\mathsf{in}\, \Lnt{t_{{\mathrm{3}}}}  \mid 
  \\ & & &  \mathsf{promote}_\kappa\, \Lnt{t'} \,\mathsf{for}\, \Lnt{t''} \,\mathsf{in}\, \Lnt{t}  \mid  \mathsf{derelict}_\kappa\, \Lnt{t} \\
\end{array}
\]

The syntax for types has been extended to include the exchange
modality, and the syntax of terms follow suit.  The terms $ \mathsf{exchange_l}\, \Lnt{t_{{\mathrm{1}}}} , \Lnt{t_{{\mathrm{2}}}} \,\mathsf{with}\, \Lmv{x} , \Lmv{y} \,\mathsf{in}\, \Lnt{t_{{\mathrm{3}}}} $ and $ \mathsf{exchange_r}\, \Lnt{t_{{\mathrm{1}}}} , \Lnt{t_{{\mathrm{2}}}} \,\mathsf{with}\, \Lmv{x} , \Lmv{y} \,\mathsf{in}\, \Lnt{t_{{\mathrm{3}}}} $
are used to explicitly track uses of exchange throughout proofs.  

We add the following typing rules to $\lambda\text{L}$:
{\small
\begin{mathpar} 
  \LdruleTXXEOne{} \and
  \LdruleTXXETwo{} \and
  \LdruleTXXEr{} \and
  \LdruleTXXEl{} 
\end{mathpar}
}
\begin{figure}[h]
  \small
  \begin{mdframed}
    \begin{mathpar}      
      \LdruleRXXBetaEDR{} \and
      \LdruleRXXNatEl{} \and
      \LdruleRXXNatEr{} 
    \end{mathpar}
  \end{mdframed}
  \caption{Rewriting Rules for The Typed Lambek Calculus: $\lambda\text{L}_\kappa$}
  \label{fig:rewrite-LE}
\end{figure}
\noindent
The reduction rules are in Figure~\ref{fig:rewrite-LE}, and are vary
similar to the rules from $\lambda\text{L}_!$.

\textbf{Strong normalization.}  Similarly, we show that we can embed
$\lambda\text{L}_\kappa$ into ILL, but the embedding is a bit more
interesting.

\begin{definition}
  \label{def:embed-lambda-L!-in-ILL}
  The following is an extension of the embedding of $\lambda\text{L}$
  into ILL resulting in an embedding of types and terms of
  $\lambda\text{L}_\kappa$ into ILL.  First, we define $ \Lsym{(}   \kappa  \Lnt{A}   \Lsym{)} ^{\mathsf{e} }  =
   !    \Lnt{A} ^{\mathsf{e} }   $, and then the following defines the embedding of
  terms:
  \begin{center} \small
    \begin{math}
      \begin{array}{rlllllllllllllllllll}        
          (   \mathsf{exchange_l}\, \Lnt{t_{{\mathrm{1}}}} , \Lnt{t_{{\mathrm{2}}}} \,\mathsf{with}\, \Lmv{x} , \Lmv{y} \,\mathsf{in}\, \Lnt{t_{{\mathrm{3}}}}   )  ^{\mathsf{e} }  & = & \Lsym{[}   \Lnt{t_{{\mathrm{2}}}} ^{\mathsf{e} }   \Lsym{/}  \Lmv{x}  \Lsym{]}  \Lsym{[}   \Lnt{t_{{\mathrm{1}}}} ^{\mathsf{e} }   \Lsym{/}  \Lmv{y}  \Lsym{]}    \Lnt{t_{{\mathrm{3}}}} ^{\mathsf{e} }  \\
          (   \mathsf{exchange_r}\, \Lnt{t_{{\mathrm{1}}}} , \Lnt{t_{{\mathrm{2}}}} \,\mathsf{with}\, \Lmv{x} , \Lmv{y} \,\mathsf{in}\, \Lnt{t_{{\mathrm{3}}}}   )  ^{\mathsf{e} }  & = & \Lsym{[}   \Lnt{t_{{\mathrm{2}}}} ^{\mathsf{e} }   \Lsym{/}  \Lmv{x}  \Lsym{]}  \Lsym{[}   \Lnt{t_{{\mathrm{1}}}} ^{\mathsf{e} }   \Lsym{/}  \Lmv{y}  \Lsym{]}    \Lnt{t_{{\mathrm{3}}}} ^{\mathsf{e} }  \\
          (   \mathsf{promote}_\kappa\, \Lnt{t'} \,\mathsf{for}\, \Lnt{t''} \,\mathsf{in}\, \Lnt{t}   )  ^{\mathsf{e} }  & = &  \mathsf{promote}_!\,  \Lnt{t'} ^{\mathsf{e} }  \,\mathsf{for}\,  \Lnt{t''} ^{\mathsf{e} }  \,\mathsf{in}\,   \Lnt{t} ^{\mathsf{e} }   \\
          (   \mathsf{derelict}_\kappa\, \Lnt{t}   )  ^{\mathsf{e} }  & = &  \mathsf{derelict}_!\,   \Lnt{t} ^{\mathsf{e} }   \\
      \end{array}
    \end{math}
  \end{center}
\end{definition}

The embedding translates the exchange modality into the of-course
modality of ILL.  We do this so as to preserve the comonadic structure
of the exchange modality.  One might think that we could simply
translate the exchange modality to the identity, but as Benton showed
\cite{Benton:1995c}, this would result in an embedding that does not
preserve reductions.  Furthermore, the left and right exchange terms
are translated away completely, but this works because ILL contains
exchange in general, and hence, does not need to be tracked
explicitly.  We now have strong normalization and confluence.

\begin{theorem}[Strong Normalization]
  \label{theorem:type_preserving_embedding_lambdaLk}
  \begin{itemize}
  \item If $ \Gamma  \vdash  \Lnt{t}  :  \Lnt{A} $ in $\lambda\text{L}_!$, then
    $  \Gamma ^{\mathsf{e} }   \vdash   \Lnt{t} ^{\mathsf{e} }   :    \Lnt{A} ^{\mathsf{e} }   $ in ILL.

  \item If $ \Lnt{t_{{\mathrm{1}}}}  \rightsquigarrow  \Lnt{t_{{\mathrm{2}}}} $ in $\lambda\text{L}_!$, then $  \Lnt{t_{{\mathrm{1}}}} ^{\mathsf{e} }   \rightsquigarrow    \Lnt{t_{{\mathrm{2}}}} ^{\mathsf{e} }   $
    in ILL.

  \item If $ \Gamma  \vdash  \Lnt{t}  :  \Lnt{A} $, then $\Lnt{t}$ is strongly normalizing.

  \item The reduction relation, $ \redto $, modulo the commuting
    conversions is confluent.
  \end{itemize}
\end{theorem}
\begin{proof}
  The first two cases hold by straightforward induction on the form of
  the assumed typing or reductions derivation.  They then imply the
  third case.
\end{proof}

\subsection{The Typed Lambek Calculus: $\lambda\text{L}_{!\kappa}$}
\label{subsec:the_typed_lambek_calculus:lambda-l-!kappa}
If we combine all three of the previous typed Lambek Calculi, then we
obtain the typed Lambek Calculus $\lambda\text{L}_{!\kappa}$.  The
main characteristics of this system are that it provides the benefits
of the non-symmetric adjoint structure of the Lambek Calculus with the
ability of having exchange, and the of-course modality, but both are
carefully tracked within the proofs.

Strong normalization for this calculus can be proved similarly to the
previous calculi by simply merging the embeddings together.  Thus,
both modalities of $\lambda\text{L}_{!\kappa}$ would merge into the
of-course modality of ILL.  The Church-Rosser property also holds for
$\lambda\text{L}_{!\kappa}$ by extending the proof of confluence for
ILL by Bierman \cite{Bierman:1994} just as we did for the other
systems.  Thus, we have the following results.

\begin{theorem}[Strong Normalization]
  \label{theorem:strong_normalization_lambdaL!k}
  If $ \Gamma  \vdash  \Lnt{t}  :  \Lnt{A} $, then $\Lnt{t}$ is strongly normalizing.
\end{theorem}

\begin{theorem}[Confluence]
  \label{theorem:confluence-lambdaL!k}
  The reduction relation, $ \redto $, modulo the commuting conversions
  is confluent.
\end{theorem}

\section{Acknowledgments}
\label{sec:acknowledgments}
The authors would like to thank the anonymous reviewers for their
feedback which did make this a better paper.  The second author was
partially supported by the NSF grant \#1565557.

\section{Conclusions}
We have recalled how to use biclosed posets and sets to construct
dialectica-like models of the Lambek Calculus. This construction is
admittedly not the easiest one, which is the reason why we use
automated tools to verify our definitions, but it has one striking
advantage. It shows how to introduce modalities to recover the
expressive power of intuitionistic (and a posteriori classical)
propositional logic to the system. We know of no other construction of
models of Lambek Calculus that does model modalities, not using their
syntactic properties. (The traditional view in algebraic semantics is
to consider idempotent operators for modalities like !). The
categorical semantics here has been described before
\cite{depaiva1991}, but the syntactic treatment of the lambda-calculi,
in the style of \cite{benton1993} had not been done and there were
doubts about its validity, given the results of Jay \cite{jay1991}. We
are glad to put this on a firm footing, using another one of Benton's
ideas: his embedding of intuitionistic linear logic into system
F. Finally, we envisage more work, along the lines of algebraic proof
theory, for modalities and non-symmetric type systems.



\end{document}